\pdfoutput=1
\documentclass[
aps,
pra,
footinbib,
reprint,
superscriptaddress,
longbibliography]{revtex4-2}

\usepackage{amsmath,amssymb,amsfonts}
\usepackage{mathtools}    
\usepackage{bm}
\usepackage{physics2}
\usephysicsmodule{braket}
\renewcommand{\ketbra}[2]{\ket{#1}\!\bra{#2}} %
\usepackage{stmaryrd}
\usepackage{mathrsfs}
\usepackage{derivative}

\usepackage{enumitem}

\usepackage{amsthm}
\theoremstyle{definition}
\newtheorem{thm}{Theorem}
\newtheorem{prop}[thm]{Proposition}

\newcommand{\R}{\mathbb{R}}
\newcommand{\C}{\mathbb{C}}

\newcommand{\Hilb}{\mathcal{H}}
\renewcommand{\P}{\mathcal{P}}
\newcommand{\CZ}{\mathbf{CZ}}
\newcommand{\G}{\mathcal{G}}

\DeclareMathOperator{\Tr}{Tr}
\DeclareMathOperator{\supp}{supp}

\renewcommand\paragraph[1]{%
  \par\emph{#1---}\kern2pt\relax\ignorespaces}

\mathchardef\standardl=\mathcode`l
\mathcode`l=\ell
\newcommand{\deactivatel}{\mathcode`l=\standardl}

\makeatletter
\edef\operator@font{\operator@font\noexpand\deactivatel}
\makeatother

\usepackage{hyperref}
\hypersetup{
  colorlinks=true,
  citecolor=magenta,
  linkcolor=blue,
  urlcolor=cyan
}

\usepackage{comment}
\usepackage{version}
\excludeversion{note} %
\includeversion{keep} %

\begin{document}
\newcommand{\nameoftitle}{%
Exact Thermal Stabilizer Eigenstates at Infinite Temperature
}

\title{\nameoftitle}

\author{Akihiro Hokkyo}
\email{hokkyo@cat.phys.s.u-tokyo.ac.jp}
\affiliation{Department of Physics, Graduate School of Science, The University of Tokyo, 7-3-1 Hongo, Bunkyo, Tokyo 113-8654, Japan}

\begin{abstract}
Understanding how microscopic few-body interactions give rise to thermal behavior in isolated quantum many-body systems remains a central challenge in nonequilibrium statistical mechanics. While individual energy eigenstates are expected to reproduce thermal equilibrium values, analytic access to highly entangled thermal eigenstates of nonintegrable Hamiltonians remains scarce.
In this Letter, we construct exact infinite-temperature eigenstates of generically nonintegrable two-body Hamiltonians using stabilizer states. These states can fully reproduce thermal expectation values for all spatially local observables, 
extending previously known Bell-pair-based constructions to a broader class. 
At the same time, we prove a sharp no-go theorem: stabilizer eigenstates of two-body Hamiltonians cannot satisfy microscopic thermal equilibrium for all four-body observables. 
This bound is tight, as we explicitly construct a translationally invariant Hamiltonian whose stabilizer eigenstate is thermal for all two-body and three-body observables as well as 
all spatially local observables.
Our results suggest that reproducing higher-order thermal correlations requires nonstabilizer degrees of freedom, providing analytic insight into the interplay between interaction locality, microscopic thermal equilibrium, and quantum computational complexity.
\end{abstract}

\maketitle
\paragraph{Introduction}
How an isolated quantum many-body system, evolving under reversible unitary dynamics, can nonetheless exhibit irreversible thermalization is a central question in quantum statistical mechanics.
Recent experiments on artificial quantum many-body systems, such as ultracold atomic gases and trapped ions, have demonstrated thermalization dynamics while remaining effective isolation from the environment~\cite{kaufmanQuantumThermalizationEntanglement2016,closTimeResolvedObservationThermalization2016}.
On the theoretical side, the eigenstate thermalization hypothesis (ETH) has emerged as a leading framework, 
which asserts that individual energy eigenstates of generic interacting Hamiltonians reproduce thermal expectation values of observables~\cite{neumannBeweisErgodensatzesUnd1929,deutschQuantumStatisticalMechanics1991,srednickiChaosQuantumThermalization1994,rigolThermalizationItsMechanism2008}.
This picture is supported by extensive numerical studies~\cite{rigolThermalizationItsMechanism2008,steinigewegPushingLimitsEigenstate2014,beugelingFinitesizeScalingEigenstate2014,kimTestingWhetherAll2014,garrisonDoesSingleEigenstate2018}.

Despite this progress, fully analytic access to thermal eigenstates of concrete and physically realistic many-body Hamiltonians remains extremely limited.
Random-matrix approaches successfully capture many aspects of thermal behavior~\cite{wignerStatisticalDistributionWidths1951,bohigasCharacterizationChaoticQuantum1984,goldsteinApproachThermalEquilibrium2010,reimannGeneralizationNeumannsApproach2015}, 
but they do not incorporate essential physical constraints such as locality and the few-body nature of interactions~\cite{hamazakiAtypicalityMostFewBody2018}.
As a result, rigorous analytic results are largely restricted to special settings, including nonthermal counterexamples such as quantum many-body scars~\cite{shiraishiSystematicConstructionCounterexamples2017,moudgalyaExactExcitedStates2018} and Hilbert-space fragmentation~\cite{salaErgodicityBreakingArising2020,khemaniLocalizationHilbertSpace2020,moudgalyaHilbertSpaceFragmentation2022}, 
as well as proofs in constrained regimes such as free systems~\cite{shiraishiNatureAbhorsVacuum2024,tasakiMacroscopicIrreversibilityQuantum2025,roosMacroscopicThermalizationHighly2025} or low temperatures~\cite{kuwaharaEigenstateThermalizationClustering2020}.

A key obstacle underlying these challenges is the intrinsic complexity of thermal pure states.
Such states exhibit volume-law entanglement, whereas many analytically and numerically tractable state families, such as matrix product states~\cite{schollwockDensitymatrixRenormalizationGroup2011}, usually obey area-law or sub-volume-law entanglement scaling.
This mismatch highlights a fundamental tension between the entanglement structure required for thermalization and that accessible within controlled analytic frameworks, motivating the search for alternative highly entangled yet exactly tractable state constructions.

Recently, progress has been made through the construction of the first explicit families of analytically tractable thermal energy eigenstates.
In particular, the entangled antipodal pair (EAP) states~\cite{chibaExactThermalEigenstates2024,ivanovVolumeentangledExactEigenstates2024b,yonetaThermalPureStates2024,mohapatraExactVolumeLawEntangled2025,mestyanCrosscapStatesTunable2025}, 
built from long-range entangled Bell-pair dimers, 
reproduce thermal expectation values at infinite temperature for all spatially local observables, providing concrete realizations of microscopic thermal equilibrium (MITE) for local subsystems~\cite{goldsteinThermalEquilibriumMacroscopic2015,goldsteinMacroscopicMicroscopicThermal2017,moriThermalizationPrethermalizationIsolated2018}.
However, these states already deviate from thermal behavior at the level of two-body observables, 
suggesting a qualitative distinction between spatially local thermal equilibrium and thermal equilibrium for few-body observables.
This therefore raises the question of whether this limitation reflects a peculiarity of such dimer-based constructions or a more general constraint imposed by few-body interactions.

In this Letter, we address this question by developing a stabilizer-based framework
for constructing analytically tractable yet highly entangled energy eigenstates
of nonintegrable many-body Hamiltonians.
Stabilizer states~\cite{gottesmanStabilizerCodesQuantum} and their subclass of graph states~\cite{heinMultipartyEntanglementGraph2004,heinEntanglementGraphStates2006}
form a well-studied family of highly entangled states in quantum computation
and provide a controlled setting in which reduced density matrices can be evaluated exactly and can exhibit volume-law entanglement.
Notably, this framework naturally incorporates previously known dimer-based EAP states and enables a systematic exploration of thermalization constraints beyond such constructions.

Our main results establish a sharp and tight no-go theorem.
We show that while stabilizer eigenstates of two-body Hamiltonians can realize microscopic thermal equilibrium at infinite temperature for all spatially local observables, they cannot satisfy microscopic thermal equilibrium for all four-body observables.
This bound is tight, as demonstrated by explicit constructions of translationally invariant nonintegrable Hamiltonians whose stabilizer eigenstates are thermal for all two-body and three-body observables, as well as for all spatially local observables.
We further identify the structural origin of this constraint by characterizing when a stabilizer state can appear as a zero-energy eigenstate, suggesting that reproducing higher-order few-body thermal correlations under few-body interactions generically requires nonstabilizer, or ``magic''~\cite{veitchResourceTheoryStabilizer2014}, degrees of freedom.

\paragraph{Local and few-body microscopic thermal equilibrium}
We consider an $N$-qubit system with sites labeled by $i\in [N]\coloneqq\{1,\dots,N\}$ and Hilbert space $\Hilb_{[N]}\coloneqq\bigotimes_{i\in[N]}\Hilb_i$ with $\Hilb_i\cong\C^2$.
For a given Hamiltonian $\hat{H}$ and a subsystem $A\subset [N]$, 
a state $\rho$ is said to be in \emph{microscopic thermal equilibrium} (MITE) on $A$
if the reduced state $\rho_A=\Tr_{[N]\setminus A}\rho$
is indistinguishable from the reduced state of a Gibbs state
$\tau_\beta\propto e^{-\beta \hat{H}}$ with the same energy expectation value~\cite{goldsteinThermalEquilibriumMacroscopic2015,goldsteinMacroscopicMicroscopicThermal2017,moriThermalizationPrethermalizationIsolated2018}.
In this Letter, we restrict attention to the infinite-temperature case ($\beta=0$),
for which this condition is equivalent to the reduced density matrix being maximally mixed.
A state is said to be in \emph{$k$-body MITE} 
if it is in MITE on every subsystem $A$ with $|A|=k$.
This means that all $k$-body observables have thermal expectation values.
Moreover, when $[N]$ is regarded as a spin chain with periodic boundary conditions,
a state is said to be in \emph{$l$-local MITE} 
if it is in MITE on all contiguous subsystems of size $l$.

\paragraph{Stabilizer eigenstates of local Hamiltonians}
We focus on stabilizer states~\cite{gottesmanStabilizerCodesQuantum}, 
including graph states~\cite{heinMultipartyEntanglementGraph2004,heinEntanglementGraphStates2006}, 
as analytically tractable yet highly entangled many-body pure states.
A subgroup $G$ of the $N$-qubit Pauli group $\P_{[N]}$
is called a stabilizer group if $G$ is abelian and does not contain $-I$.
A state $\ket{\psi_G}$ is called a stabilizer state associated with $G$
if it is the simultaneous eigenstate with eigenvalue $+1$
of all elements of $G$.
Throughout this Letter, 
$G$ is assumed to be maximal, which implies that
$\ket{\psi_G}$ is unique up to a phase factor. 
Such states admit efficient classical descriptions, 
allowing for exact evaluation
of reduced density matrices~\cite{gottesmanStabilizerCodesQuantum,fattalEntanglementStabilizerFormalism2004,heinMultipartyEntanglementGraph2004,heinEntanglementGraphStates2006}.
We consider zero-energy stabilizer eigenstates of two-body Hamiltonians, 
i.e., Hamiltonians that are sums of terms acting nontrivially on at most two sites.
We fix the additive constant of the Hamiltonian so that it is traceless,
in which case the zero energy corresponds to 
the infinite-temperature Gibbs state,
i.e., the maximally mixed state on the full Hilbert space.

As an illustrative example, 
the EAP states can be represented as stabilizer states.
For even $N=2N'$, consider a stabilizer that is generated by two-body operators of the form 
\begin{equation}
  \{\sigma_i^x\sigma_{i+N'}^x,\ \sigma_i^z\sigma_{i+N'}^z\mid 1\le i\le N'\}.
\end{equation}
The corresponding stabilizer state is a translationally symmetric EAP state $\ket{\mathrm{EAP}(I)}$, which appeared in Refs.~\cite{chibaExactThermalEigenstates2024,yonetaThermalPureStates2024}. 

\paragraph{Main results}
Our main results establish a sharp constraint on the few-body MITE property
for stabilizer eigenstates of two-body Hamiltonians.
First, we prove the following no-go theorem.
\begin{thm}[No-go]\label{thm:no-go}
Any zero-energy stabilizer eigenstate of a nonzero two-body Hamiltonian
is not in $k$-body MITE for any $k\ge 4$.
\end{thm}
Intuitively, stabilizer states have strictly constrained correlation structure fixed by a stabilizer group, 
and the two-body parent-Hamiltonian condition forces four-body reduced states to retain stabilizer signatures that cannot be fully erased.

We further show that this bound is tight.
\begin{thm}[Achievability and tightness]\label{thm:tight}
We regard $[N]$ as a spin chain with the periodic boundary condition.
There exists a two-body, finite-range, translationally invariant Hamiltonian (Eq.~\eqref{eq:Hamiltonian} below) admitting a stabilizer eigenstate with zero energy that is in three-body MITE and $O(N)$-local MITE.
\end{thm}
Together, Theorems~\ref{thm:no-go} and~\ref{thm:tight} show that,
among stabilizer eigenstates of two-body Hamiltonians at infinite temperature,
the largest integer $k$ for which $k$-body MITE can be achieved is $k=3$.
This result shows that stabilizer states go beyond existing dimer-based constructions
by reproducing thermal equilibrium not only for one-body but also for two-body and three-body observables,
while simultaneously exposing a stringent constraint imposed by few-body interactions. 
It further suggests that, under the physical constraint of few-body interactions, capturing higher-order correlations requires nonstabilizerness, 
or \emph{magic}~\cite{veitchResourceTheoryStabilizer2014}, 
as possessed by typical many-body quantum states~\cite{turkeshiPauliSpectrumNonstabilizerness2025}.

\paragraph{Explicit construction saturating the bound.}
We next present an explicit example of a two-body Hamiltonian whose zero-energy
stabilizer eigenstate saturates the bound of Theorem~\ref{thm:no-go}.
Let $N\ge8$ be even, and consider the following Hamiltonian with the periodic boundary condition:
\begin{equation}
    \hat{H}=J\sum_{i\in[N]}(\sigma^z_i\sigma^z_{i+3}-\sigma^x_i\sigma^x_{i+1}),
    \label{eq:Hamiltonian}
\end{equation}
where $J\in\R$ is an arbitrary real constant.
Due to the presence of the medium-range interaction
$\sigma^z_i\sigma^z_{i+3}$, this Hamiltonian is expected to be nonintegrable.
We have numerically verified that 
the distribution of energy-level spacings~\cite{bohigasCharacterizationChaoticQuantum1984,oganesyanLocalizationInteractingFermions2007,atasDistributionRatioConsecutive2013,dalessioQuantumChaosEigenstate2016} of the Hamiltonian~\eqref{eq:Hamiltonian} 
follows that of the Gaussian orthogonal ensemble (see End Matter), 
providing strong evidence of nonintegrability.

\begin{figure}[tbp]
    \centering
    \includegraphics[width=0.7\linewidth]{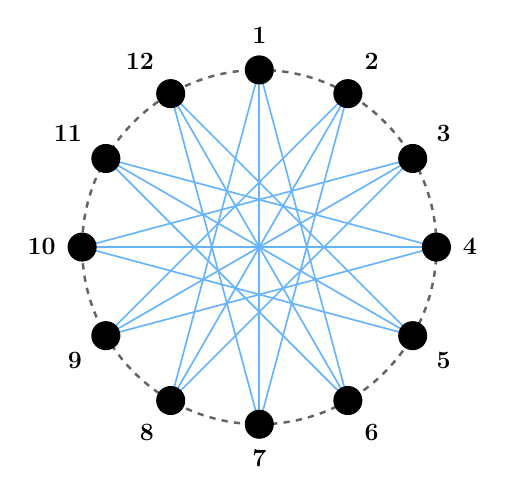}
    \caption{Schematic illustration of the graph $\G_1$ (Eq.~\eqref{eq:edge_graph}) 
    for $N=12$.
    Each vertex $i$ is connected to three vertices
    $i+N/2-1$, $i+N/2$, and $i+N/2+1$.
    The corresponding graph state is in three-body MITE and also in five-local MITE. }
    \label{fig:graph}
\end{figure}

This Hamiltonian admits as an eigenstate the graph state $\ket{\G_1}$
associated with a graph $\G_1$ (Fig.~\ref{fig:graph}) whose edge set is given by
\begin{equation}
    E_{\G_1}=\{\{i,i+N/2-1\},\{i,i+N/2\}\mid i\in [N]\}.
    \label{eq:edge_graph}
\end{equation}
Explicitly, the graph state $\ket{\G_1}$ is given by
\begin{equation}
    \ket{\G_1}=\prod_{\{i,j\}\in E_{\G_1}}\CZ_{ij}
    \bigotimes_{i\in[N]}\ket{+}_i,
\end{equation}
where $\CZ_{ij}$ denotes the controlled-$Z$ gate acting on sites $i$ and $j$,
and $\ket{+}$ is the eigenstate of $\sigma^x$ with eigenvalue $+1$.

We now show that $\ket{\G_1}$ is a zero-energy eigenstate of
Eq.~\eqref{eq:Hamiltonian}.
By general properties of graph states~\cite{heinEntanglementGraphStates2006}, 
the following operators are stabilizers
of $\ket{\G_1}$:
\begin{equation}
    K_{\G_1}^{(i)}=
    \sigma^x_i
    \sigma^z_{i+N/2-1}
    \sigma^z_{i+N/2}
    \sigma^z_{i+N/2+1}.
\end{equation}
Therefore,
$K_{\G_1}^{(i)}K_{\G_1}^{(i+1)}
=\sigma^x_i\sigma^x_{i+1}
 \sigma^z_{i+N/2-1}\sigma^z_{i+N/2+2}$ 
 is a stabilizer of $\ket{\G_1}$.
Thus, we obtain
\begin{align}
    &\sigma^x_i\sigma^x_{i+1}
     \sigma^z_{i+N/2-1}\sigma^z_{i+N/2+2}
     \ket{\G_1}
     =\ket{\G_1}\nonumber\\
    \Leftrightarrow\;
    &(\sigma^z_{i+N/2-1}\sigma^z_{i+N/2+2}
      -\sigma^x_i\sigma^x_{i+1})
      \ket{\G_1}=0.
      \label{eq:factorized_annihilation}
\end{align}
Summing this identity over $i$ yields $\hat{H}\ket{\G_1}=0$.

Finally, we verify that this state satisfies the MITE condition.
General properties of graph states~\cite{heinEntanglementGraphStates2006} yield 
the following criterion: 
$\ket{\G_1}$ is in MITE on a subsystem $A\subset[N]$ 
if and only if, for every nonempty subset $B\subset A$,
there exists a site $i\notin A$ such that
\begin{equation}
    |\{j\in B\mid \{i,j\}\in E_{\G_1}\}|
    \equiv 1 \pmod{2}.
    \label{eq:cond_MITE_graph}
\end{equation}
Using this criterion, we find that for $N\ge 8$,
the state $\ket{\G_1}$ is in three-body MITE and also in $(N/2-1)$-local MITE
(see End Matter).
On the other hand, since the expectation value of $K_{\G_1}^{(i)}$ is equal to $1$,
in contrast to its infinite-temperature average of $0$,
the state is not in four-body MITE.
This is consistent with Theorem~\ref{thm:no-go}. 

As prior work on constructing eigenstates using stabilizers,
we note the existence of nonthermal eigenstates known as stabilizer scars~\cite{hartseStabilizerScars2025}.
In contrast to those constructions, the states considered here exhibit thermal behavior at the level of all spatially local observables, as well as all two-body and three-body observables, 
highlighting a clear distinction between our thermal stabilizer eigenstates and previously studied nonthermal stabilizer scar states.

\paragraph{Origin of the no-go constraint}
We outline the mechanism underlying Theorem~\ref{thm:no-go}.
The key ingredient is the following proposition, 
which generalizes the construction
of Hamiltonians having EAP statesas eigenstates~\cite{chibaExactThermalEigenstates2024} and to the state discussed
in the previous section, 
and provides a necessary and sufficient condition
for the existence of a zero-energy stabilizer eigenstate.
\begin{prop}\label{prop:general_cond}
Let $\ket{\psi_G}$ be a stabilizer state with stabilizer group $G$.
A traceless Hamiltonian $\hat{H}$ on $\Hilb_{[N]}$ admits $\ket{\psi_G}$
as a zero-energy eigenstate, i.e., $\hat{H}\ket{\psi_G}=0$,
if and only if $\hat{H}$ can be written as
\begin{equation}
  \hat{H}
  =\sum_{g\in G}
    \sum_{\substack{P,Q\in\P_{[N]}^1\setminus\{I\}\\ a_{P,Q}PQ=g}}
    c(P,Q)\bigl(P-a_{P,Q}Q\bigr),
  \label{eq:Hamiltonian_from_stabilizer}
\end{equation}
where $\P_{[N]}^1\subset\P_{[N]}$ denotes the set of Pauli strings with prefactor $1$,
$a_{P,Q}\in\{\pm1,\pm i\}$ is the uniquely determined phase from $P$ and $Q$,
and the coefficients $c(P,Q)$ can be chosen arbitrarily
\footnote{
When $a_{P,Q}\in\{\pm i\}$, the coefficients must be chosen appropriately
to ensure that $\hat{H}$ is Hermitian.
}.
\end{prop}
\begin{proof}
We first prove sufficiency.
Suppose that $g\in G$ factorizes as $g=a_{P,Q}PQ$ for
$P,Q\in\P_{[N]}^1\setminus\{I\}$ and $a_{P,Q}\in\{\pm1,\pm i\}$.
Since $g\ket{\psi_G}=\ket{\psi_G}$ and $P^2=I$, we have
\begin{equation}
    a_{P,Q}Q\ket{\psi_G}=P\ket{\psi_G}.
\end{equation}
Hence, each term in Eq.~\eqref{eq:Hamiltonian_from_stabilizer}
annihilates $\ket{\psi_G}$.
Note that this argument is a straightforward generalization of the proof
of Eq.~\eqref{eq:factorized_annihilation}.

We next prove necessity.
Any traceless Hamiltonian admits a unique expansion
\begin{equation}
    \hat{H}
    =\sum_{P\in\P_{[N]}^1\setminus\{I\}} h(P)\,P
\end{equation}
with suitable coefficients $h(P)$.
From the condition $\hat{H}\ket{\psi_G}=0$, it follows that for any
$P\in\P_{[N]}^1$,
\begin{equation}
    0
    =\braket[3]{\psi_G}{P\hat{H}}{\psi_G}
    =\sum_{Q\in\P_{[N]}^1\setminus\{I\}}
      h(Q)\braket[3]{\psi_G}{PQ}{\psi_G}.
\end{equation}
Using the standard representation of a stabilizer state~\cite{fattalEntanglementStabilizerFormalism2004,heinEntanglementGraphStates2006}
\begin{equation}
    \ketbra{\psi_G}{\psi_G}
    =\frac{1}{2^N}\sum_{g\in G} g,\label{eq:stabilizer_average}
\end{equation}
we obtain
\begin{align}
    0
    &=\sum_{Q\in\P_{[N]}^1\setminus\{I\}}
      h(Q)\braket[3]{\psi_G}{PQ}{\psi_G}\nonumber\\
    &=\sum_{g\in G}
      \sum_{\substack{Q\in\P_{[N]}^1\setminus\{I\}\\ a_{P,Q}PQ=g}}
      h(Q)\,a_{P,Q}^{-1},
    \label{eq:key_equality}
\end{align}
which holds for all $P\in\P_{[N]}^1$.
From this relation, the existence of coefficients $c(P,Q)$ satisfying
Eq.~\eqref{eq:Hamiltonian_from_stabilizer} follows
(see End Matter).
\end{proof}

Note that for generic choices of the coefficients $\{c(P,Q)\}_{P,Q}$,
the Hamiltonian $\hat{H}$ is expected to be nonintegrable,
since the factorized Pauli monomials $P$ need not commute with each other,
even though all elements of the stabilizer group $G$ commute.

Proposition~\ref{prop:general_cond} shows that 
admitting a stabilizer eigenstate severely restricts the operator structure of $\hat{H}$: 
Each term must come from factorizing a stabilizer element into two Pauli operators, 
which directly obstructs reproducing generic four-body thermal correlations.
In fact, the no-go theorem follows immediately.
\begin{proof}[Proof of Theorem~\ref{thm:no-go}]
By Eq.~\eqref{eq:stabilizer_average}, 
a stabilizer state is in $k$-body MITE if and only if~\cite{heinEntanglementGraphStates2006,vandennestGraphStatesGround2008}
\begin{equation}
    k < \min_{g\in G\setminus\{I\}} |\supp(g)|\eqqcolon\delta(G),
    \label{eq:condition_fb_MITE}
\end{equation}
where $\supp(g)\subset[N]$ denotes the set of sites on which $g$
acts nontrivially.
On the other hand, if $P,Q\in\P_{[N]}^1$ satisfy $PQ\propto g$, then
\begin{equation}
    |\supp(g)| = |\supp(PQ)|
    \le |\supp(P)| + |\supp(Q)|.
\end{equation}
Therefore, whenever Eq.~\eqref{eq:condition_fb_MITE} holds, we must have
\begin{equation}
    \max\bigl(|\supp(P)|,|\supp(Q)|\bigr)
    \ge \left\lceil\frac{k+1}{2}\right\rceil.
\end{equation}
This implies that the Hamiltonian in Eq.~\eqref{eq:Hamiltonian_from_stabilizer}
contains terms acting on at least
$\left\lceil\frac{k+1}{2}\right\rceil$ sites.
Since $\hat{H}$ is assumed to be two-body,
Eq.~\eqref{eq:condition_fb_MITE} cannot hold for $k\ge 4$.
Hence, the state cannot be in $k$-body MITE for $k\ge 4$.
\end{proof}
This proof also shows more generally that, in order for stabilizer eigenstates
at infinite temperature to be in $k$-body MITE,
interactions acting on at least $\left\lceil(k+1)/2\right\rceil$ sites are necessary.
A related constraint appears in the study of graph ground states,
where $(k+1)$-body interactions are required to realize $k$-body MITE~\cite{vandennestGraphStatesGround2008}.

We also note that if one relaxes the requirement that the state be an eigenstate of a few-body Hamiltonian, 
states satisfying higher-order MITE properties can indeed exist.
At infinite temperature, the notion of $k$-body MITE coincides with the concept of $k$-uniformity~\cite{scottMultipartiteEntanglementQuantumerrorcorrecting2004}, 
and
it is known that $\lfloor 0.19N \rfloor$-uniform stabilizer states exist, 
whereas $\lceil 0.38N \rceil$-uniform stabilizer states do not exist for large $N$~\cite{arnaudExploringPureQuantum2013}.
Our results demonstrate that imposing the additional constraint of being an eigenstate of a few-body Hamiltonian is genuinely restrictive, 
suggesting that the few-body nature of interactions substantially increases the degree of nonstabilizerness (or magic) required to realize higher-order MITE.

\paragraph{Another example: star-shaped cluster states}
Using Proposition~\ref{prop:general_cond}, one can construct Hamiltonians
admitting a wide variety of stabilizer states as zero-energy eigenstates,
beyond the example in Eq.~\eqref{eq:Hamiltonian}.
As another illustration, we introduce a thermal eigenstate
constructed from cluster states~\cite{briegelPersistentEntanglementArrays2001}.

A cluster state is a paradigmatic example of a highly entangled stabilizer state,
originally introduced in the context of measurement-based quantum computation~\cite{raussendorfOneWayQuantumComputer2001,raussendorfMeasurementbasedQuantumComputation2003}.
It also plays an important role in quantum error correction and condensed-matter physics,
where it appears as a prototypical example of symmetry-protected topological phases~\cite{elseSymmetryProtectedPhasesMeasurementBased2012,sonTopologicalOrder1D2012}.
In one dimension, the cluster state can be realized as a graph state
associated with the one-dimensional cycle graph
$E_\G=\{\{i,i+1\}\mid i\in[N]\}$.
Since this state has stabilizers acting on three consecutive sites, 
it is not in $3$-local MITE.

\begin{figure}[tbp]
    \centering
    \includegraphics[width=0.7\linewidth]{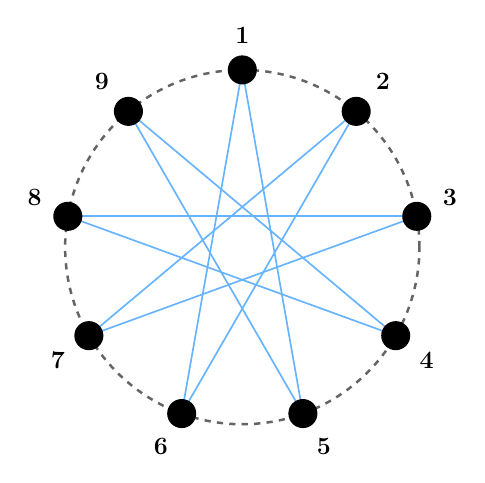}
    \caption{Schematic illustration of the graph $\G_2$ (Eq.~\eqref{eq:edge_graph_cluster})
    for $N=9$. 
    Each vertex is connected to the vertex at distance $(N-1)/2=4$.
    The corresponding graph state is in two-body MITE and also in four-local MITE. }
    \label{fig:star_cluster}
\end{figure}

However, by considering a graph that is isomorphic to the cycle graph via a nonlocal permutation of sites
(Fig.~\ref{fig:star_cluster}),
\begin{equation}
    E_{\G_2}
    =\{\{i,i+(N-1)/2\}\mid i\in[N]\},
    \label{eq:edge_graph_cluster}
\end{equation}
one can construct a graph state $\ket{\G_2}$ that is in local MITE.
Here we assume that $N$ is odd and $N\ge5$.
Owing to the structure of the underlying graph,
we refer to this state as a \emph{star-shaped cluster state}.
As in the previous section, 
applying the criterion~\eqref{eq:cond_MITE_graph}
shows that this state is in two-body MITE and also in $(N-1)/2$-local MITE.
Since cluster states admit 
an exact matrix product state representation~\cite{schonSequentialGenerationEntangled2005},
this example demonstrates that 
matrix product states can also represent complex thermal equilibrium states 
when combined with nonlocal permutations of lattice sites.

By applying the construction in Eq.~\eqref{eq:Hamiltonian_from_stabilizer},
one finds that a two-body, translationally invariant Hamiltonian
that admits the star-shaped cluster state $\ket{\G_2}$ as a zero-energy eigenstate can be written as
\begin{align}
    \hat{H}
    =\sum_{i\in [N]}
    \Bigl[
    J_1(\sigma^{z}_i\sigma^{z}_{i+2}
        &-\sigma^{x}_i\sigma^{x}_{i+1})
    +
    J_2(\sigma^{z}_i\sigma^{z}_{i+1}
        -\sigma^x_i)\nonumber\\
    &+
    J_3(\sigma^{y}_i\sigma^{z}_{i+1}
        -\sigma^{z}_i\sigma^{y}_{i+1})
    \Bigr],
    \label{eq:Hamiltonian_star_cluster}
\end{align}
where $J_1,J_2,J_3\in\R$ are arbitrary real coefficients.
For $J_1\neq 0$, this Hamiltonian is expected to be generically nonintegrable,
because even for the case of $J_3=0$ 
it has been shown that this Hamiltonian has no nontrivial local conserved quantities~\cite{shiraishiCompleteClassificationIntegrability2025}.

\paragraph{Conclusion and outlook}
We have constructed exact infinite-temperature thermal eigenstates of two-body spin-1/2 Hamiltonians using stabilizer states. 
We also established a sharp no-go theorem: no zero-energy stabilizer eigenstate of a nontrivial two-body Hamiltonian can satisfy four-body MITE. This bound is tight, as demonstrated by a translationally invariant two-body model admitting a stabilizer eigenstate thermal for all two-body and three-body observables and even for $O(N)$-local subsystems.
Our results provide controlled analytic examples relevant to ETH-like behavior
beyond dimer constructions, and clarify the possibilities and limitations of realizing thermal eigenstates within stabilizer and graph-state frameworks.

Several natural extensions of the present work merit further investigation.
One direction is to generalize our analysis to stabilizer states of higher-spin
systems~\cite{zhouQuantumComputationBased2003,hostensStabilizerStatesClifford2005}, 
where richer stabilizer structures and interaction patterns may arise.
Another important question is whether our construction can similarly be used to construct thermal pure states at finite temperature, as is the case for the EAP states~\cite{chibaExactThermalEigenstates2024,yonetaThermalPureStates2024}, 
rather than being limited to the infinite-temperature regime considered here.
Related progress on generalizations of EAP states using dimers that are not fully entangled~\cite{mohapatraExactVolumeLawEntangled2025,mestyanCrosscapStatesTunable2025}
may also be relevant in this direction.

More broadly, our results suggest the existence of a fundamental interplay
between the few-body nature of interactions, the degree of few-body microscopic
thermal equilibrium that can be achieved, and the amount of nonstabilizerness, or magic.
Clarifying this trade-off quantitatively may shed light on the structural
distinction between stabilizer states and thermal many-body quantum states 
arising from physically realistic Hamiltonians.

\paragraph{Acknowledgments}
The author thanks Balázs Pozsgay and Niklas Mueller for their valuable comments on the manuscript.
This work was supported by KAKENHI Grant
No. JP25KJ0833 from the Japan Society for the
Promotion of Science (JSPS) and FoPM, a WINGS Program, 
the University of Tokyo.
The author also acknowledges support from JSR Fellowship, the University of Tokyo and the ASPIRE program (Grant No. JPMJAP25A1) by Japan Science and Technology Agency (JST). 

\paragraph{Note added}
During the completion of this manuscript,
the author became aware of an independent related work by Dooley~\cite{dooleyParentHamiltoniansStabilizer2026}.
Where the two works overlap, 
specifically, 
in the sufficient condition for constructing stabilizer zero-energy eigenstates, 
the results are consistent.

\bibliography{bibliography,suppl}

\section*{END MATTER}
\subsection*{Energy-level spacings of Eq.~\eqref{eq:Hamiltonian}}
Numerical results for the energy-level spacing statistics of the Hamiltonian in Eq.~\eqref{eq:Hamiltonian} are shown in Fig.~\ref{fig:level_dist}.
Specifically, we employ the ratio of consecutive level spacings as a diagnostic~\cite{oganesyanLocalizationInteractingFermions2007,atasDistributionRatioConsecutive2013}.
The Hamiltonian possesses the following discrete symmetries:
\begin{itemize}
\item lattice translation $T$;
\item spatial inversion $P$;
\item $\pi$ rotation about the $x$ axis, $P_X \coloneqq \bigotimes_i X_i$;
\item $\pi$ rotation about the $z$ axis, $P_Z \coloneqq \bigotimes_i Z_i$.
\end{itemize}
In the numerical calculations, 
we focus on the symmetry sector defined by $T = P = P_X = P_Z = 1$ for system size $N = 22$.
We note that when $N$ is a multiple of $4$ (resp.~$6$), 
this sector contains additional sublattice-induced antisymmetries (resp.~symmetries).

Figure~\ref{fig:level_dist} shows that the Hamiltonian in Eq.~\eqref{eq:Hamiltonian} exhibits random-matrix-like spectral statistics, strongly suggesting its nonintegrability.

\begin{figure}[htbp]
    \centering
    \includegraphics[width=\linewidth]{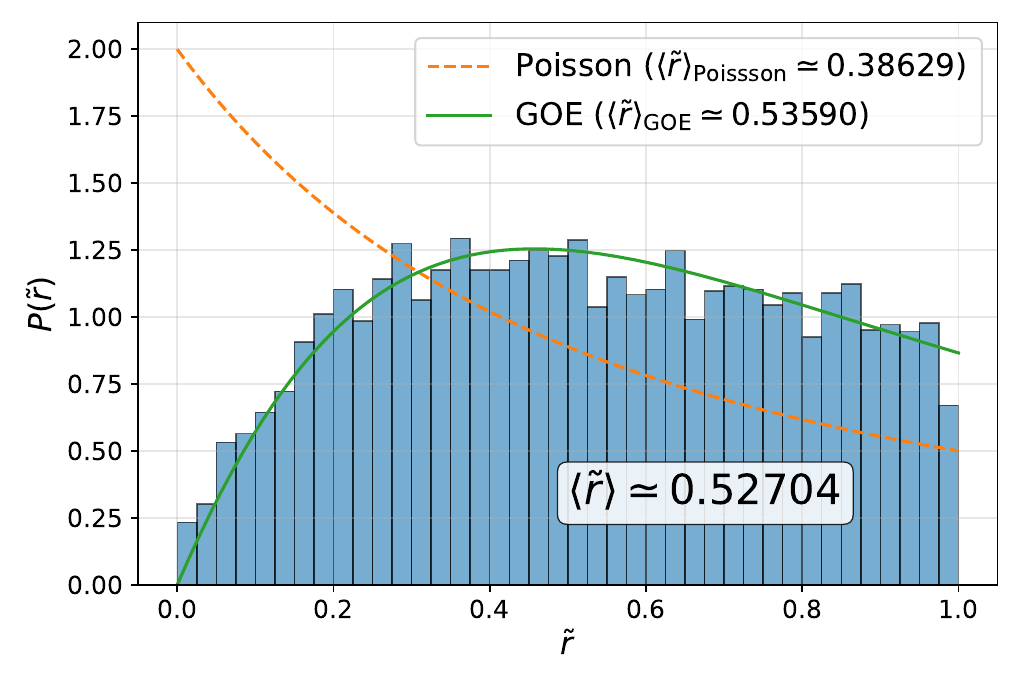}
    \caption{The distribution $P(\tilde r)$ is shown, 
    together with its mean value $\langle \tilde r \rangle$, 
    where $\tilde r = \min\{r, r^{-1}\} \in [0,1]$ and $r$ denotes the ratio of consecutive level spacings~\cite{oganesyanLocalizationInteractingFermions2007,atasDistributionRatioConsecutive2013}.
    The histogram is obtained from exact diagonalization of Eq.~\eqref{eq:Hamiltonian}, 
    using the central half of the eigenvalues 
    in the sector with $N=22$ and $T = P = P_X = P_Z = 1$.
    The dotted and solid curves show the exact results for Poisson statistics and the Gaussian orthogonal ensemble (GOE), respectively~\cite{atasDistributionRatioConsecutive2013}.
    TPoisson statistics characterize integrable systems,
    while GOE statistics describe nonintegrable systems with time-reversal symmetry~\cite{berryLevelClusteringRegular1977,bohigasCharacterizationChaoticQuantum1984}.
    The level statistics of Eq.~\eqref{eq:Hamiltonian} are well described by those of the GOE.
}
\label{fig:level_dist}
\end{figure}

\subsection*{MITE property of graph states}

We prove the local and few-body MITE properties of $\ket{\G_1}$ and $\ket{\G_2}$.

First, we show that $\ket{\G_1}$ is in MITE on
$A=\{1,\dots,N/2-1\}$ using Eq.~\eqref{eq:cond_MITE_graph}.
Take an arbitrary nonempty subset $B\subset A$ and let $b_0=\min B$.
Define $i=b_0+N/2-1$.
Since $N/2\le i\le N-2$, we have $i\in [N]\setminus A$.
Moreover, among the sites in $B$, the only site connected to $i$ is $b_0$.
Therefore, this $i$ satisfies the condition in Eq.~\eqref{eq:cond_MITE_graph}.
By an entirely analogous argument to the one above, 
one finds that $\ket{\G_2}$ is also in MITE on
$A=\{1,\dots,(N-1)/2\}$.

Next, we establish the few-body MITE properties.
We first show that $\ket{\G_2}$ is in MITE on $\{1,n\}$ for any $n\in[N]$.
Since $\ket{\G_2}$ has already been shown to be $(N-1)/2$-local MITE,
it suffices to consider the case $n=(N+1)/2$.
In this case, the site $(N+3)/2$ is connected to $1$ but not to $(N+1)/2$,
whereas the site $N$ is connected to $(N+1)/2$ but not to $1$.
Hence, Eq.~\eqref{eq:cond_MITE_graph} holds for any nonempty subset
$B\subset\{1,(N+1)/2\}$ with a suitable choice of $i$.

Finally, we show that $\ket{\G_1}$ is in MITE on
$A=\{1,m,n\}$ for any three distinct sites $1<m<n$ in $[N]$.
Take an arbitrary nonempty subset $B\subset A$.
\begin{itemize}
    \item If $|B|=1$, then $b\in B$ is connected to exactly three sites.
    Therefore, there exists at least one site not contained in $A$
    that is connected to $b$.
    Choosing such a site as $i$, Eq.~\eqref{eq:cond_MITE_graph} is satisfied.
    
    \item If $|B|=2$, then there are at most two sites connected to both elements of $B$.
    Moreover, if the two sites in $B$ are connected to each other,
    there is no site connected to both of them.
    As a result, there exist at least two sites that are not in $B$
    and are connected to exactly one of the two sites in $B$.
    Consequently, among the sites with this property,
    there exists at least one site that is not contained in $A$.
    Choosing such a site as $i$, Eq.~\eqref{eq:cond_MITE_graph} is satisfied.
    
    \item If $|B|=3$, namely $B=A$,
    consider first the case $A=\{1,2,3\}$.
    Then the site $i=N/2+1$ is connected to all three sites,
    and thus satisfies Eq.~\eqref{eq:cond_MITE_graph}.
    Therefore, we may assume $m\ge 3$.
    When $m\neq N/2$, there always exist two consecutive sites
    $j_1,j_1+1$ that are connected only to $1$ (but not to $m$),
    and two consecutive sites $j_m,j_m+1$ that are connected only to $m$
    (but not to $1$).
    If $n$ coincides with any of these sites,
    then since consecutive sites are not connected to each other,
    one of $i\in\{n\pm1\}$ satisfies Eq.~\eqref{eq:cond_MITE_graph}.
    If none of $j_1,j_1+1,j_m,j_m+1$ coincides with $n$,
    then there exists at least one site among them that is not connected to $n$.
    Choosing such a site as $i$ again satisfies Eq.~\eqref{eq:cond_MITE_graph}.
    The case $m=N/2$ can be treated in the same manner as above.
\end{itemize}
From the above arguments, we conclude that $\ket{\G_1}$ is in three-body MITE.

\subsection*{Complete proof of Proposition~\ref{prop:general_cond}}

Assume that Eq.~\eqref{eq:key_equality} holds. We show that one can construct coefficients $c(P,Q)$ satisfying Eq.~\eqref{eq:Hamiltonian_from_stabilizer}.
We introduce a relation on $\P_{[N]}^1$ by
\begin{equation}
    P\sim Q \;\Leftrightarrow\; \exists\, a_{P,Q}\ \text{such that } a_{P,Q}PQ\in G .
\end{equation}
Since $G$ is a group, this relation is an equivalence relation.
We define $a_{P,Q}=0$ for $P\not\sim Q$. 
Note that $a_{P,Q}a_{Q,P}=1$ if $P\sim Q$. 

For $P\in\P_{[N]}^1\setminus\{I\}$, we define
\begin{equation}
    n_P=\bigl|\{Q\in \P_{[N]}^1\setminus\{I\}\mid P\sim Q\}\bigr| .
\end{equation}
Note that for $P,Q\in\P_{[N]}^1\setminus\{I\}$,
\begin{equation}
    P\sim Q \;\Rightarrow\; n_P=n_Q .
    \label{eq:lem1}
\end{equation}

For $P,Q\in\P_{[N]}^1\setminus\{I\}$, we define
\begin{equation}
    c(P,Q)\coloneqq -\frac{1}{n_P}a_{Q,P}h(Q) .
    \label{eq:lem2}
\end{equation}
From Eq.~\eqref{eq:key_equality}, it follows that
\begin{align}
    \sum_{\substack{Q\in\P_{[N]}^1\setminus\{I\}\\ P\sim Q}}c(P,Q)
    &=-\frac{1}{n_P}\sum_{g\in G}
      \sum_{\substack{Q\in\P_{[N]}^1\setminus\{I\}\\ a_{P,Q}PQ=g}}
      h(Q)\,a_{P,Q}^{-1}\nonumber\\
      &=0 .
      \label{eq:lem3}
\end{align}

With the coefficients $c(P,Q)$ defined above, 
we can expand the right-hand side of Eq.~\eqref{eq:Hamiltonian_from_stabilizer} as
\begin{equation}
    \sum_{\substack{P,Q\in\P_{[N]}^1\setminus\{I\}\\ P\sim Q}}
    c(P,Q)\bigl(P-a_{P,Q}Q\bigr)
    =
    \sum_{R\in\P_{[N]}^1\setminus\{I\}}
    h'(R)R,
\end{equation}
where $h'(R)$ is given by
\begin{equation}
    h'(R)=\sum_{\substack{Q\in\P_{[N]}^1\setminus\{I\}\\ R\sim Q}}
    \bigl(c(R,Q)-c(Q,R)a_{Q,R}\bigr) .
\end{equation}
Using Eqs.~\eqref{eq:lem1}, \eqref{eq:lem2}, and \eqref{eq:lem3}, we obtain
\begin{align}
    h'(R)
    &=h(R)\sum_{\substack{Q\in\P_{[N]}^1\setminus\{I\}\\ R\sim Q}}
      \frac{1}{n_Q}a_{R,Q}a_{Q,R}\nonumber\\
    &=h(R)\frac{1}{n_R}
      \sum_{\substack{Q\in\P_{[N]}^1\setminus\{I\}\\ R\sim Q}}1\nonumber\\
    &=h(R) .
\end{align}
Therefore, Eq.~\eqref{eq:Hamiltonian_from_stabilizer} is satisfied, 
completing the proof.
\end{document}